\documentclass[11pt,reqno]{amsproc}

\usepackage{txfonts,array,amsmath,amssymb,amsthm,color,vmargin,overpic,amstext}
\usepackage{amsmath,amsthm,amsfonts,amscd,amssymb,amstext}
\usepackage{mathtools}
\usepackage{booktabs} 
\usepackage[T1]{fontenc}
\usepackage{relsize}
\usepackage{graphicx,color}
\usepackage{tabularx}
\usepackage{titletoc}
\usepackage{caption}
\usepackage{subcaption}
\usepackage[colorlinks=true,linkcolor=black,citecolor=black,urlcolor=blue]{hyperref}
\usepackage{listings}
\usepackage{url}
\usepackage{alltt}
\usepackage{bbm}
\usepackage{hyperref}
\usepackage{fancyhdr}
\newtheorem{theorem}{Theorem}

\newtheorem{remark}[theorem]{Remark}

\newtheorem{lemma}[theorem]{Lemma}
\newtheorem{proposition}[theorem]{Proposition}
\newtheorem{corollary}[theorem]{Corollary}

\numberwithin{equation}{section}
\numberwithin{theorem}{section}
\usepackage[foot]{amsaddr}
\makeatletter\let\@wraptoccontribs\wraptoccontribs\makeatother
\title{Characteristic polynomials of random truncations: moments, duality and asymptotics}

\author{Alexander Serebryakov and Nick Simm*\\\\
\textit{With an Appendix by Guillaume Dubach$^\dagger$}\footnote{$^{\dagger}$ IST Austria, Am Campus 1, 3400 Klosterneuburg, Austria. E-mail: \textit{guillaume.dubach@ist.ac.at.}}}
\address{*Department of Mathematics, University of Sussex, Brighton, BN1 9RH, United Kingdom.}
\email{a.serebryakov@sussex.ac.uk, n.j.simm@sussex.ac.uk}

\begin{document}

\begin{abstract}
We study moments of characteristic polynomials of truncated Haar distributed matrices from the three classical compact groups $\mathrm{O(N)}$, $\mathrm{U(N)}$ and $\mathrm{Sp(2N)}$. For finite matrix size we calculate the moments in terms of hypergeometric functions of matrix argument and give explicit integral representations highlighting the duality between the moment and the matrix size as well as the duality between the orthogonal and symplectic cases. Asymptotic expansions in strong and weak non-unitarity regimes are obtained. Using the connection to matrix hypergeometric functions, we establish limit theorems for the log-modulus of the characteristic polynomial evaluated on the unit circle.
\end{abstract}
\maketitle

\section{Introduction and main results}
Characteristic polynomials of random matrices have been the subject of intense research for a few decades with several applications ranging from QCD \cite{Stephanov_1996,Akemann:2004aa}, quantum chaos and disordered systems \cite{Andreev_Simons,haake2010quantum} and equilibria of complex systems \cite{FK16}. They are also mathematically rich objects having been connected to the distribution of zeros of the Riemann Zeta function \cite{Keating_Snaith_2000} and other $L$-functions \cite{Keating:2000ub,Conrey:2003tk} with study of their moments playing a prominent role. Recently the correspondence of random matrices with logarithmically correlated Gaussian fields and Gaussian multiplicative chaos measures generated a resurgence of interest in characteristic polynomials of random matrices, see \cite{FHK12,FK14}, the recent survey \cite{Bailey:2021uz} and references therein. This correspondence often relies on exact formulae for moments or correlation functions of characteristic polynomials, together with a precise asymptotic analysis. Significant progress was made in recent years in the one dimensional setting (\textit{e.g.} Hermitian ensembles) where such averages are most readily available \cite{Webb2014TheCP,Berestycki_Webb_Wong_2018, Forkel:2021us}. 

In this paper we study moments of characteristic polynomials for a class of non-Hermitian random matrices known as \textit{truncations}, introduced in the works \cite{Khoruzhenko_2010,Zyczkowski:2000ue} and defined as follows. Let $\mathrm{O(N)}$, $\mathrm{U(N)}$ and $\mathrm{Sp(2N)}$ denote the three classical compact groups of $N \times N$ orthogonal, unitary or symplectic matrices. These matrix groups each come with a unique translation invariant measure known as Haar measure, see \cite{Meckes_book_2019} for a comprehensive treatise. Choosing an element $U$ from one of $\mathrm{O(N)}$, $\mathrm{U(N)}$ or $\mathrm{Sp(2N)}$ with respect to the Haar measure, we consider the sub-block decomposition
\begin{align}
	U 
	= 
	\begin{pmatrix}
		A && B\\
		C && D	
	\end{pmatrix}
	, \label{uconstruct}
\end{align}
where the principal sub-matrix $A$ is of size $M \times M$ with $M < N$. By invariance of the Haar measure we consider the matrix $A$ without loss of generality and say that $A$ belongs to the truncated orthogonal, unitary or symplectic ensemble. We use the standard Dyson notation $\beta=1,2$ or $4$ to distinguish the orthogonal, unitary or symplectic cases. The decomposition \eqref{uconstruct} arises very frequently in mesoscopic physics where the various sub-blocks are used to calculate reflection and transmission properties of an open quantum system with given symmetries \cite{DBB10}. Mathematically, the truncations are related to two well-studied ensembles of random matrices that arise in particular limits, as follows. After proper rescaling, in the limit $N \to \infty$ with $M$ fixed it is known that the entries of $A$ are well approximated by i.i.d. standard Gaussian random variables \cite{Diaconis:1992aa}. Gaussian matrices of this type with no additional symmetries are known as \textit{Ginibre ensembles} \cite{Ginibre_65} and may be regarded as the archetype of non-Hermitian random matrices. At the other extreme, if $N-M$ is of order $1$ then we may expect spectral characteristics similar to random matrices drawn from the classical compact groups with Haar measure.

Our goal is to obtain exact formulae and asymptotic expansions for moments of the type
\begin{align}
R^{(\beta)}_{2k}(x) =
\begin{cases}
\begin{aligned}
&\mathbb{E}
		\left [
			\det ( xI_{M} - A)^{2k}
		\right ]
		& \text{for $\beta = 1$,}\\[1ex]
&		\mathbb{E}
		\left [
			|\det ( xI_{M} - A)|^{2k}
		\right ]
		& \text{for $\beta = 2$,} \\[1ex]
&    \mathbb{E}
		\left [
			\det ( xI_{M} - A)^{k}
		\right ] & \text{for $\beta = 4$},
	\end{aligned} 
\end{cases}
\label{R_definition}
\end{align}
where $x \in \mathbb{C}$ and $I_{M}$ is the $M \times M$ identity matrix. We emphasize that when $\beta=4$, we assume throughout that $x$ is a real quaternion and interpret $\det(xI_{M}-A)$ using the standard representation of quaternions as $2 \times 2$ complex valued matrices. Typically we shall assume that the exponent $k$ is a positive integer, though in certain situations we also consider non-integer moments and write $k = \frac{\gamma}{2}$ where $\gamma$ is a general parameter. This way of defining $R^{(\beta)}_{2k}(x)$ allows us to present unified formulae for each symmetry index $\beta \in \{1,2,4\}$.

In order to state our results, we define
\begin{equation}
	\beta' = \frac{4}{\beta}, \label{betap}
\end{equation}
and the Vandermonde determinant
\begin{equation}
\Delta(\vec{t}) = \prod_{1 \leq i < j \leq k}(t_{j}-t_{i}) = \det\bigg\{t_{i}^{j-1}\bigg\}_{i,j=1}^{k}.
\end{equation}

\begin{theorem} \label{theorem_truncated_duality_integrals_sigma}
	Consider the averages defined by \eqref{R_definition}. Then for any $k \in \mathbb{N}$ and $\beta \in \{2,4\}$, we have

	\begin{align}
		&R^{(\beta)}_{2k}(x)
		=
		\frac{1}{S_{k,N}(\beta)}
		\int_{[0,1]^k}
		\prod_{i=1}^{k}dt_{i}\,t_i^{N-M}
		\left ( 1 + (|x|^{2} - 1)t_i \right )^M|\Delta(\vec{t})|^{\beta'}
		\label{truncation_characteristic_poly},
	\end{align}
	where
	\begin{equation}
	S_{k,N}(\beta)= \prod_{j=0}^{k-1}\frac{\Gamma(N+1+j\frac{2}{\beta})\Gamma(1+j\frac{2}{\beta})\Gamma(1+(j+1)\frac{2}{\beta})}{\Gamma(N+2+(k+j-1)\frac{2}{\beta})\Gamma(1+\frac{2}{\beta})}. \label{skn}
	\end{equation}
If $\beta=1$ then \eqref{truncation_characteristic_poly} holds with the replacement $|x|^{2} \to x^{2}$.
\end{theorem}

The integral representation \eqref{truncation_characteristic_poly} may be regarded as a duality that replaces the initial average over $M \times M$ matrices with one over a fixed dimensionality $k$, with $M$ appearing only as a parameter. Hence the roles of $M$ and $k$ are interchanged. A second aspect of this duality occurs between the real and symplectic ensembles, as is indicated by the interchanged power of repulsion $\beta'$. Such dualities have been observed in several settings in random matrix theory, most frequently in the case of Hermitian or Circular ensembles \cite{Brezin:2000uh,Brezin_2001aa_Real, FK04,FS09, Desrosiers:2009wu,Desrosiers:2014aa,DL15}. For non-Hermitian ensembles, related dualities are known in the complex case \cite{AV03,Fyodorov:2007un, Deano_Simm_2020}, especially for the complex Ginibre ensemble \cite{Nishigaki_2002, Kan05}.

The proof of Theorem \ref{theorem_truncated_duality_integrals_sigma} is based on the theory of hypergeometric functions of matrix argument. Forrester and Rains \cite{Forrester_2009} applied these techniques to obtain moments of characteristic polynomials of the Ginibre ensembles; here we follow a similar approach in the present context of the truncations $A$ defined in \eqref{uconstruct}. The starting point is a generalised binomial expansion for powers of the characteristic polynomial appearing in the averages \eqref{R_definition}; these expansions are given as sums over partitions involving Schur polynomials. One is then tasked with evaluating Schur polynomial averages in the eigenvalues of the random truncations. By combining results of \cite{Forrester_2006, Forrester_2009} with knowledge of the Kaneko-Kadell integral for the average of a Jack function \cite{Kaneko_1993, Kadell1997TheSS}, we recognise the remaining sum over partitions as a $_2F_1$ hypergeometric function of matrix argument. The same approach also applies to a more general average where the matrix $A$ in definition \eqref{R_definition} is replaced with $AV$ for some deterministic $M \times M$ matrix $V$, see Theorems \ref{theorem_average_as_2F1} and \ref{propv}. We then arrive at Theorem \ref{theorem_truncated_duality_integrals_sigma} as a particular case by exploiting an integral representation of $_2F_1$ due to Kaneko \cite{Kaneko_1993}. 

If we evaluate the hypergeometric function at special values we obtain other types of exact identities, one of which we highlight below.\begin{theorem}
\label{prop:introbndry}
Let $\beta=2$, so that $A$ is an $M \times M$ truncation of an $N \times N$ Haar distributed unitary matrix. On the boundary $x = e^{i\theta}$ we have the following exact evaluation
\begin{equation}
\mathbb{E}
		\left [
			|\det ( e^{i\theta}I_{M} - A)|^{\gamma}		\right ] = \prod_{j=N-M+1}^{N}
	\frac{\Gamma(j) \Gamma(j+\gamma)}{ \left ( \Gamma(j + \frac{\gamma}{2}) \right )^2}, \label{bndrymoms}
\end{equation}
valid for $\mathrm{Re}(\gamma) > -1$.
\end{theorem}
In the particular case $M=N$ the identity \eqref{bndrymoms} specialises to the explicit $\mathrm{U(N)}$ moments obtained by Keating and Snaith and famously used to conjecture moments of the Riemann zeta function \cite{Keating_Snaith_2000}. Thus Theorem \ref{prop:introbndry} provides a sub-unitary $1$-parameter extension of such moments. In Corollary \ref{bndrycor} we obtain similar results for $\beta=1$ and $\beta=4$; we will show that these are related to another investigation of Keating and Snaith \cite{Keating:2000ub} concerning the classical groups $\mathrm{SO(2N)}$ and $\mathrm{Sp(2N)}$ associated with other symmetry classes of $L$-functions.

We now turn to the asymptotic analysis of the moments for large matrix size. This depends on the relative growth of the dimensions $M$ and $N$ entering the definition of $A$ in \eqref{uconstruct}. We will consider the limiting regimes discussed in the works \cite{Zyczkowski:2000ue,Khoruzhenko_2010}, see also the survey \cite{FS03}. The regime of \textit{strong non-unitarity} is defined by the assumption that $M$ and $N$ are proportional: specifically we assume that $N \equiv N_{M}$ is a sequence of positive integers such that $N_{M} \to \infty$ as $M \to \infty$ with the ratio $\mu := \frac{M}{N} \to \tilde{\mu} $ as $M \to \infty$, where $\tilde{\mu} \in (0,1)$. In this regime the eigenvalues of the sub-block $A$ are distributed on the disc of radius $\sqrt{\tilde{\mu}}$ \cite{Zyczkowski:2000ue,Khoruzhenko_2010,Meckes:2019vf}. 

To state our result, recall that the \textit{Gaussian $\beta$-Ensemble} denotes $k \times k$ real symmetric ($\beta=1$), complex Hermitian ($\beta=2$) or quaternion self-dual ($\beta=4$) random matrices $H$ with probability density function proportional to $\mathrm{exp}\left(-\frac{1}{2}\mathrm{Tr}(H^{2})\right)$.
\begin{theorem}
\label{theorem_strong_non-unitarity_uniform_asymptotics}
Consider the averages defined by \eqref{R_definition}, with $x$ real if $\beta=1$. Then in the regime of strong non-unitarity with $\mu := \frac{M}{N} \to \tilde{\mu} \in (0,1)$ as $M \to \infty$, the following asymptotics hold uniformly for $x$ varying on the disc of radius $\sqrt{\mu}$,
\begin{equation}
\begin{split}
	R^{(\beta)}_{2k}(x)
	=	
	&
	M^{\frac{k^{2}}{\beta} + \frac{k}{2}(1-\frac{2}{\beta} )}\mu^{Mk}
	\left(\frac{1 - \mu}{1-|x|^{2}}\right)^{Mk(1 - \mu^{-1})}\left ( \frac{\sqrt{1 - \mu}}{1-|x|^2} \right )^{k + \frac{2}{\beta}k(k-1)}\\
	&\times
	(2\pi)^{\frac{k}{2}}\prod_{j=0}^{k-1}
	\frac{1}{\Gamma(1 + \frac{2}{\beta}j)}\,
	\mathbb{P}
	\left ( 
		\lambda_{\mathrm{max}}^{(\mathrm{G\beta'E})} < \sqrt{M} \frac{\mu - |x|^{2}}{\mu \sqrt{1 - \mu}}
	\right )
	(1 + o(1)), \qquad M \to \infty, \label{gbetaeresult}
	\end{split}
\end{equation}
where $\lambda_{\mathrm{max}}^{(\mathrm{G\beta'E})}$ is the largest eigenvalue in the $k \times k$ Gaussian $\beta'$-Ensemble, with $\beta'$ given by \eqref{betap}.
\end{theorem}
We point out that the $\mathrm{G\beta' E}$ probability in \eqref{gbetaeresult} is primarily there to describe the boundary asymptotics of $R^{(\beta)}_{2k}(x)$ characterised by values of $|x|$ in a small vicinity of $\sqrt{\mu}$, and gives a non-trivial contribution to the constant term of the asymptotics only within a $M^{-\frac{1}{2}}$ neighbourhood of this point. For values of $x$ strictly inside the disc of radius $\sqrt{\mu}$ this probability can be replaced by $1$ without affecting the asymptotic expansion. 

In contrast, the \textit{regime of weak non-unitarity} is characterised by instead fixing $\kappa = N-M$ so that only a finite number of rows and columns are truncated in the construction \eqref{uconstruct}. In this regime the eigenvalues of the sub-block $A$ lie within a $O(M^{-1})$ neighbourhood of the unit circle. Now a different classical ensemble of Hermitian matrices arises in the asymptotics, the \textit{Laguerre $\beta$-Ensemble}. These are $k \times k$ real symmetric ($\beta=1$), complex Hermitian ($\beta=2$) or quaternion self-dual ($\beta=4$) positive definite random matrices $H$ with probability measure proportional to $|\det(H)|^{\alpha}\mathrm{exp}\left(-\mathrm{Tr}(H)\right)$ with a parameter $\alpha>-1$. They are also commonly referred to as \textit{Wishart matrices}, see \cite{Muirhead_Aspects, Forrester_2010_Log} for further background.
\begin{theorem} \label{theorem_dulaity_integrals_asymptotics}
Consider the averages defined by \eqref{R_definition} with $x$ real if $\beta=1$ and set $|x|^{2}=1-\frac{2u}{M}$ with $u>0$ fixed. Then for $\kappa = N - M$ fixed we have
	\begin{equation}
	\begin{split}
		R^{(\beta)}_{2k}(x)
		=
		\left(\frac{M}{2u}\right)^{\frac{2}{\beta}k^2 + k(1 - \frac{2}{\beta})}
	&\left (
	\prod^{k-1}_{j=0}
	\frac{\Gamma(\kappa + 1 +\frac{2}{\beta}j)}{\Gamma(1+\frac{2}{\beta}j)}
	\right )\\
	&\times
	(2u)^{- k\kappa}
	\mathbb{P}
		\left ( 
			\lambda_{\mathrm{max}}^{(\mathrm{L\beta'E}_{\kappa})} < 2 u
		\right )
	(1 + o(1)), \qquad M \to \infty, \label{lbetaeresult}
	\end{split}
	\end{equation}
	where $\lambda_{\mathrm{max}}^{(\mathrm{L\beta'E}_{\kappa})}$ is the largest eigenvalue in the $k \times k$ Laguerre $\beta'$-Ensemble with parameter $\kappa$, with $\beta'$ given by \eqref{betap}.
\end{theorem}

\begin{remark}
In the case $\beta' = \beta =2$, the identity \eqref{truncation_characteristic_poly} is a particular case of \cite[Theorem 1]{Fyodorov:2007un}, see also \cite{Deano_Simm_2020} where the $\beta=2$ asymptotics of Theorem \ref{theorem_dulaity_integrals_asymptotics} were obtained. For $\beta=2$ our approach is comparable to the one of \cite{Fyodorov:2007un}, though here we extend the duality to the cases $\beta \in \{1,4\}$ and also obtain non-integer moments of the characteristic polynomial, see the later identities \eqref{char_pol_pre_integral} and \eqref{char_pol_integral}. To our knowledge, the asymptotics of Theorem \ref{theorem_strong_non-unitarity_uniform_asymptotics} are new for all three $\beta \in \{1,2,4\}$, although there exist analogues for the Ginibre ensembles, for $\beta=2$ in \cite{webb_wong_2019,Deano_Simm_2020}, and for $\beta=1$ in \cite{Wang:2021tg}.
\end{remark} 
\begin{remark}
In \cite{Deano_Simm_2020} it was pointed out that for $\beta=2$ the distribution functions arising in \eqref{gbetaeresult} and \eqref{lbetaeresult} are expressible in terms of Painlev\'e transcendents; solutions of certain integrable non-linear second order differential equations. In the cases $\beta \in \{1,4\}$ there is also a connection to integrable systems; the so-called Pfaff-KP hierarchy that can be used to write down non-linear differential-difference equations associated with the distribution functions \cite{AvM01,MS13}. 
\end{remark}
Finally, we present central limit theorems for the log-modulus of the characteristic polynomial of $A$, specialised such that the matrix hypergeometric theory mentioned above gives explicit results. Limit theorems for random determinants have a long history, for a recent account, see the discussion in \cite[Sec. 1.2]{BBK21}. They have also been the subject of considerable recent interest \textit{e.g.} \cite{Bourgade:2018wm, Borgo:2019tt,NV14,Chen:2017uw,BBK21}.  
\begin{theorem} \label{limit_theorem_logdet_at_I}
	Let $A$ be a real ($\beta = 1$), complex ($\beta = 2$) or real quaternion ($\beta = 4$) $M \times M$ truncation of a Haar distributed matrix taken from $\mathrm{O(N)}, \mathrm{U(N)}, \mathrm{Sp(2N)}$. In the weak non-unitarity regime where we fix $\kappa = N - M$ and $\theta \in \mathbb{R}$, we have the convergence in distribution to a standard normal random variable
	\begin{align}
		\frac{\log | \det (e^{i \theta}I_{M}-A)|- e_{\beta} \log M }{\sqrt{v_{\beta}\log M}}	
		\overset{d}{\underset{M\to \infty}{\longrightarrow}}
		\mathcal{N}(0,1)
		, \label{N01}
	\end{align}
	where if $\beta=1$ we assume that $\theta \in \{0,\pi\}$. The coefficients in \eqref{N01} are
	\begin{align}
		&e_{1}
		=
		-
		\frac{1}{2}
		,
		\quad
		e_{2}
		=
		0
		,
		\quad
		e_{4}
		=
		\frac{1}{2}, \label{es}
	\end{align}
	and
\begin{equation}
v_{1} = 1, \quad v_{2}=\frac{1}{2}, \quad v_{4} = 1. \label{vs}
\end{equation}
On the other hand, in the strong non-unitarity regime $\mu := \frac{M}{N} \to \tilde{\mu} $ as $M \to \infty$, we have the convergence in distribution
	\begin{align}
		\log | \det (e^{i \theta}I_{M}-A)|
		\overset{d}{\underset{M\to \infty}{\longrightarrow}}
		\mathcal{N}\left(m_{\tilde{\mu}}, \sigma^{2}_{\tilde{\mu}}\right)	
		, \label{Nmv}
	\end{align}
with mean and variance $m_{\tilde{\mu}} = -e_{\beta}\log(1-\tilde{\mu})$ and $\sigma^{2}_{\tilde{\mu}} = -v_{\beta}\log(1-\tilde{\mu})$.
\end{theorem}
It is notable above that in the strong non-unitarity regime we obtain convergence to a limiting Gaussian variable with no normalization, while the weak non-unitary regime requires a normalization by the familiar $\sqrt{\log M}$ factor that also appears for random matrices from the classical compact groups \cite{Keating_Snaith_2000, Keating:2000ub}. If we consider the limiting distribution when the argument of the characteristic polynomial is at the origin, rather than on the unit circle, we find that even the limiting Gaussianity is not guaranteed. 
\begin{theorem} \label{theorem_log_determiants_limit_theorem}
	Let $A$ be a real ($\beta = 1$), complex ($\beta = 2$) or real quaternion ($\beta = 4$) $M \times M$ truncation of a Haar distributed matrix taken from $\mathrm{O(N)}, \mathrm{U(N)}, \mathrm{Sp(2N)}$. Then for a fixed $\kappa = N - M$, we have the convergence in distribution
	\begin{align}
		\log | \det (A) |
		+
		\frac{\kappa}{2}	
		\log \frac{\beta}{2} M
		\overset{d}{\underset{M\to \infty}{\longrightarrow}}
		\frac{1}{2}\sum^{\kappa-1}_{j=0} \log \Gamma^{(\beta)}_j \label{gam}
		,
	\end{align}
	where $\Gamma^{(\beta)}_j$ are independent Gamma distributed random variables with parameter $\frac{\beta}{2} + \frac{\beta}{2} j$.
\end{theorem}
To our knowledge, limit theorems for determinants of random truncations were only considered fairly recently, in the special case of $\log|\det(A)|$ with $\beta=2$ in the regime of strong non-unitarity \cite{Chen:2017uw}. In contrast to \eqref{gam} this regime gives rise to a Gaussian limit, see Proposition \ref{theorem_logdet_snu_limit}. In fact, the latter problem can be easily Hermitised allowing for the application of earlier results on Hermitian random matrices \cite{R07}. In contrast, our limit theorems for $\log|\det(e^{i\theta}I_{M}-A)|$ seem to be less obvious and rely instead on the aforementioned relation to matrix hypergeometric functions.\\

The structure of this paper is as follows. In Section \ref{section_Schur_function_averages} we calculate averages of Schur polynomials in the eigenvalues of the random truncations. This is used to identify the moments of characteristic polynomials in terms of a matrix hypergeometric function, obtaining Theorems \ref{theorem_truncated_duality_integrals_sigma} and \ref{prop:introbndry} as particular cases. In Section \ref{sec:asympt} we turn to the asymptotic analysis and prove Theorems \ref{theorem_strong_non-unitarity_uniform_asymptotics} and \ref{theorem_dulaity_integrals_asymptotics}. In Section \ref{section_limit_theorems} we obtain exact distributional identities for the log-determinants and use these to prove Theorems \ref{limit_theorem_logdet_at_I} and \ref{theorem_log_determiants_limit_theorem}. Finally, we present an Appendix written by Guillaume Dubach giving a direct proof of Theorem \ref{prop:introbndry} and its extension to the orthogonal and symplectic groups, see Corollary \ref{bndrycor}.

\section*{Acknowledgements}
Both authors gratefully acknowledge financial support of the Royal Society, grant URF\textbackslash R1\textbackslash 180707. We would like to thank Emma Bailey, Yan Fyodorov and Jordan Stoyanov for helpful comments on an earlier version of this paper.

\section{Schur polynomial averages and matrix hypergeometric functions}
\label{section_Schur_function_averages}
The goal of this Section is to evaluate the averages \eqref{R_definition} in terms of hypergeometric functions of matrix argument and prove Theorem \ref{theorem_truncated_duality_integrals_sigma}. This requires the use of tools from the theory of symmetric functions and we now introduce the relevant notation, see also \cite{FS09, Forrester_2009}.

A partition of a non-negative integer $k \in \mathbb{N}$ is a weakly decreasing ordered sequence $\nu = (\nu_1,\nu_2,..)$ of non-negative integers with only finitely many non-zero terms such that $k = |\nu| = \sum_i \nu_i$. The number of non-zero terms $\nu_{i}$ is the length of the partition and each non-zero term is called a part. The dominance partial order on the set of partitions is defined in the following way. For partitions $\mu = (\mu_1,\mu_2,\ldots)$ and $\nu = (\nu_1,\nu_2,\ldots)$ we say that $\mu \geq \nu$ if and only if $|\mu| = |\nu|$ and $\mu_{1}+\ldots+\mu_{i} \geq \nu_{1} + \ldots + \nu_{i}$ for all $i \geq 1$. 

A function $f$ of $M$ variables is said to be symmetric if for each $\sigma \in S_M$, the group of permutations on $M$ symbols, we have $f(x_1,\ldots,x_M) = f(x_{\sigma(1)},\ldots,x_{\sigma(M)})$. For a matrix $X$ with eigenvalues $x_{1},\ldots,x_{M}$ we define $f(X) := f(x_1,\ldots,x_{M})$. The monomial symmetric functions are given by the sum of monomials
\begin{equation}
m_{\nu}(X) = \sum_{\sigma}x_{1}^{\nu_{\sigma(1)}}\ldots x_{M}^{\nu_{\sigma(M)}}
\end{equation}
over all distinct permutations of $(\nu_1,\ldots,\nu_{M})$, where $\nu$ has length $\leq M$. A particular kind of symmetric function we will need are those known as the Jack polynomials, which are defined in terms of a partition $\nu$ and a parameter $\alpha>0$. This parameter will be related to the Dyson index $\beta$ of the given matrix ensemble by the relation $\alpha = \frac{2}{\beta}$. The Jack polynomials $P_\nu^{(\alpha)}(X)$ can be defined as eigenfunctions of the differential operator 
\begin{align}
	\sum^M_{j=1}
	\left (
		x_j \frac{\partial}{\partial x_j}
	\right )^2
	+
	\frac{M-1}{\alpha}	
	\sum^M_{j=1}
	x_j \frac{\partial}{\partial x_j}
	+
	\frac{2}{\alpha}
	\sum_{1\leq j < k \leq M}
	\frac{x_j x_k}{x_j - x_k}
	\left (
		\frac{\partial}{\partial x_j}
		-
		\frac{\partial}{\partial x_k}
	\right ), \label{op}
\end{align}
with the additional structure
\begin{equation}
P_{\nu}^{(\alpha)}(X) = m_{\nu}(X)+\sum_{\mu < \nu}a^{(\alpha)}_{\mu\nu}m_{\mu}(X),
\end{equation}
where $m_{\mu}$ are the monomial symmetric functions, $a^{(\alpha)}_{\mu\nu}$ are coefficients that do not depend on $M$ and $\mu < \nu$ is the dominance ordering on partitions, see \cite{Stanley1989SomeCP} for further details. When $\alpha=1$ these are precisely the Schur polynomials,
\begin{equation}
s_\nu(X) =\frac{\det\bigg\{x_j^{\nu_k+M-k}\bigg\}_{j,k=1}^{M}}{\det\bigg\{x_j^{M-k}\bigg\}_{j,k=1}^{M}}.
\end{equation}
It will be convenient to define the Jack polynomials with a slightly different normalization. For this purpose, we introduce the generalised Pochhammer symbol,
\begin{align}
	[u]^{(\alpha)}_\nu	 
	= 
	\prod^M_{j=1}
	\frac{\Gamma(u - (j-1)/ \alpha + \nu_j)}{\Gamma(u - (j-1)/ \alpha )}
	\label{generalised_pochhammer}
\end{align}
and the quantities
\begin{align}
	d'_\nu	
	=
	\frac{\alpha^{|\nu|} [(M-1) / \alpha + 1]^{(\alpha)}_\nu }{ \bar{f}^{1/\alpha}(\nu) }
	,
\end{align}
and
\begin{align}
	\bar{f}^{\> 1/\alpha}(\nu)
	=
	\prod_{1 \leq i < j \leq M}
	\frac{(1+(j-i-1)/\alpha + \nu_i - \nu_j)_{1/\alpha}}{(1+(j-i-1)/\alpha)_{1/\alpha}}	
	,
\end{align}
where $(u)_n = \Gamma(u+n)/\Gamma(u)$ is the classical Pochhammer symbol. Then the normalised Jack polynomials are
\begin{equation}
C_\nu^{(\alpha)}(X) = \frac{\alpha^{|\nu|} |\nu|!}{d'_\nu} P_{\nu}^{(\alpha)}(X). \label{njack}
\end{equation}

For a partition $\nu$ we define $2\nu$ as the partition obtained by doubling each part of $\nu$, and $\nu^{2}$ the partition obtained by repeating each part of $\nu$ twice.
\begin{lemma}
\label{theorem_Forrester_Rains}
	Suppose the distribution of $X$ is invariant under $X\mapsto UX$, $X\mapsto XU$ with $U \in \mathrm{O(N)}, \mathrm{U(N)}$ or $\mathrm{Sp(2N)}$ respectively. Then for real matrices,
\begin{align}
	\mathbb{E}
		\left [
			s_\mu (VX)
		\right ]
	=
	\begin{cases}
		\displaystyle{
		\frac{C_\nu^{(2)}(VV^T)}{\left ( C_\nu^{(2)}(1^N) \right )^2}
		\mathbb{E}
			\left [
				C_\nu^{(2)}(XX^T)
			\right ]
		}
		,  
		&\mbox{if } \mu = 2 \nu   \\ 
		0, &\mbox{otherwise;}
	\end{cases}
\end{align}
for complex matrices,
\begin{align}
	\mathbb{E}
		\left [
			s_\mu (VX)
			s_\nu (X^\dag V^\dag)
		\right ]
	=
		\delta_{\mu, \nu}
		\frac{C_\nu^{(1)}(VV^\dag)}{\left ( C_\nu^{(1)}(1^N) \right )^2}
		\mathbb{E}
			\left [
				C_\nu^{(1)}(XX^\dag)
			\right ]
		,
\end{align}
and for real quaternion matrices,
\begin{align}
	\mathbb{E}
		\left [
			s_\mu (VX)
		\right ]
	=
	\begin{cases}
	\displaystyle{
		\frac{C_\nu^{(1/2)}(VV^\dag)}{\left ( C_\nu^{(1/2)}(1^N) \right )^2}
	\mathbb{E}
		\left [
			C_\nu^{(1/2)}(XX^\dag)
		\right ]
		}
		,  &\mbox{if } \mu = \nu^2 \\ 
		0, &\mbox{otherwise.}
	\end{cases}
\end{align}
\end{lemma}
The above result was obtained in \cite{Forrester_2009}. Its algebraic origins are described in detail in the book of Macdonald \cite[Chapter VII]{Macdonald} where each result in Lemma \ref{theorem_Forrester_Rains} corresponds to the \textit{zonal spherical function} for an underlying Gelfand pair associated to each classical group. 

\begin{theorem}
	\label{theorem_schur_average_pre_sum}
	Let $A$ be an $M \times M$ truncation of a Haar distributed matrix from classical compact groups $\mathrm{O(N)}, \mathrm{U(N)}$ or $\mathrm{Sp(2N)}$ with $N \geq 2M$. Then for $\beta=1$,
	\begin{align}
	\mathbb{E}
		\left [
			s_\mu (VA)
		\right ]
	=
	\begin{cases}
		\displaystyle{
		\frac{d'_{2\nu}|_{\alpha = 1}}{ \left | \nu \right |! 2^{2\left | \nu \right |}} 
		\frac{1}{\left [ \frac{1}{2} N \right ]^{(2)}_\nu}
		C_\nu^{(2)}(VV^T) 
		}
		,  
		&\mbox{if } \mu = 2 \nu   \\ 
		0, &\mbox{otherwise;}
	\end{cases}
	\label{schur_average_pre_sum_beta_1}
\end{align}
for $\beta=2$,
\begin{align}
	\mathbb{E}
		\left [
			s_\mu (VA)
			s_\nu (A^\dag V^\dag)
		\right ]
	=
		\delta_{\mu, \nu}
		\frac{ \left ( d'_{\nu}|_{\alpha = 1} \right )^2 }{ \left | \nu \right |! } 
		\frac{1}{\left [ N \right ]^{(1)}_\nu}
		C_\nu^{(1)}(VV^\dag)
		,
	\label{schur_average_pre_sum_beta_2}
\end{align}
and for $\beta=4$, with $V$ a real quaternion matrix,
\begin{align}
	\mathbb{E}
		\left [
			s_\mu (VA)
		\right ]
	=
	\begin{cases}
	\displaystyle{
		\frac{d'_{\nu^2}|_{\alpha = 1}}{ \left | \nu \right |! } 
		\frac{1}{\left [ 2 N \right ]^{(1/2)}_\nu} 
		C_\nu^{(1/2)}(VV^\dag)
		}
		,  &\mbox{if } \mu = \nu^2 \\ 
		0, &\mbox{otherwise.}
	\end{cases}
	\label{schur_average_pre_sum_beta_4}
\end{align}
\end{theorem}

\begin{proof} 
By invariance of the Haar measure we immediately see that Theorem \ref{theorem_Forrester_Rains} applies to the matrix $A$ defined in \eqref{uconstruct}. The required Schur polynomial average is thus reduced to averaging a Jack polynomial in the eigenvalues of $AA^{\dagger}$. The joint probability density function of the eigenvalues $y_{1},\ldots,y_{M}$ of $Y = AA^{\dagger}$ was computed in \cite{Forrester_2006} and is given by the \textit{Jacobi $\beta$-Ensemble}, 
\begin{align}
	\frac{1}{S_{M}(a,b,\beta)}\,\prod^M_{j=1}
	y_j^{a}
	(1 - y_j)^{b}
	|\Delta(\vec{y})|^{\beta}
	, \label{ForJPDF}
\end{align}
where the parameters are $a = \frac{\beta}{2}-1$ and $b = \frac{\beta}{2}(N-2M+1)-1$. The normalization constant in \eqref{ForJPDF} is given by Selberg's integral \cite{Forrester08theimportance},
\begin{align}
S_{M}(a,b,\beta) &= \int_{[0,1]^{M}}\prod_{j=1}^{M}dy_{j}\,y_{j}^{a}(1-y_{j})^{b}|\Delta(\vec{y})|^{\beta}\\
&= \prod_{j=0}^{M-1}\frac{\Gamma(a+1+j\frac{\beta}{2})\Gamma(b+1+j\frac{\beta}{2})\Gamma(1+(j+1)\frac{\beta}{2})}{\Gamma(a+b+2+(M+j-1)\frac{\beta}{2})\Gamma(1+\frac{\beta}{2})}. \label{selberg}
\end{align}
Then the sought Jack polynomial average with respect to the density \eqref{ForJPDF} is the Kaneko-Kadell integral, this was evaluated explicitly in the works \cite{Kaneko_1993, Kadell1997TheSS}. Using these results we have
\begin{align}
 	\mathbb{E}
		\left [
			C_\nu^{( 2 / \beta )}(AA^\dag)
		\right ]
&= \frac{1}{S_{M}(a,b,\beta)}\,
	\int_{[0,1]^{M}}
	\prod^M_{j=1}dy_{j}\,
	y_j^{a}
	(1 - y_j)^{b}
	C_\nu^{( 2 / \beta )}(\vec{y})
	|\Delta(\vec{y})|^{\beta}
	\label{zonal_average_over_truncated_pd_2/beta}\\
	&= \frac{\left [a+1+(M-1)\frac{\beta}{2}\right ]^{(\frac{2}{\beta})}_\nu}{\left [a+b+2+2(M-1)\frac{\beta}{2} \right ]^{(\frac{2}{\beta})}_\nu}
	C_\nu^{( 2 / \beta )}(1^M)\\
	&= \frac{\left [ \frac{\beta}{2}M \right ]^{(\frac{2}{\beta})}_\nu}{\left [ \frac{\beta}{2}N \right ]^{(\frac{2}{\beta})}_\nu}
	C_\nu^{( 2 / \beta )}(1^M). \label{kaneko-integral}
 \end{align}

To complete the proof, we make use of the following evaluation of normalised Jack polynomials \cite{Borodin_2003},
\begin{align}
	&\frac{\left [ \frac{\beta}{2}M \right ]^{(2/\beta)}_\nu}{C^{(2/\beta)}_\nu(1^M)}
	=
	\frac{d'_{2\nu} |_{\alpha=1}}{|\nu|! 2^{2|\nu|}} \label{borid}
	,\end{align}
	for $\beta=1, 2$ or $4$ respectively. Substituting \eqref{borid} into \eqref{kaneko-integral} and applying Theorem \ref{theorem_Forrester_Rains} completes the proof.
\end{proof}

Our next goal is to use Theorem \ref{theorem_schur_average_pre_sum} to evaluate the moments \eqref{R_definition} in terms of hypergeometric functions of matrix argument. These are defined as
\begin{align}
	_pF^{(\alpha)}_q(a_1,...,a_p, b_1,...,b_q;X)
	=
	\sum_\nu
	\frac{1}{|\nu|!}
	\frac{[a_1]_\nu^{(\alpha)} \cdot \cdot \cdot [a_p]_\nu^{(\alpha)}}{[b_1]_\nu^{(\alpha)} \cdot \cdot \cdot [b_q]_\nu^{(\alpha)}} C^{(\alpha)}_\nu(X),
	\label{pFq_definition}
\end{align}
where the sum extends over all partitions $\nu$. For a comprehensive reference to the properties of such functions we refer the reader to the treatise by Yan \cite{Y92}. In the particular case $p=1, \ q=0$ we have the generalised binomial series \cite[Theorem 3.1]{Y92},
\begin{align}
	_1F_0^{(\alpha)}(a; X)
	=
	\sum_\nu
	\frac{[a]_\nu^{(\alpha)}}{|\nu|!} C_\nu^{(\alpha)}(X)
	=
	\det (I - X)^{-a}.
	\label{1F0}
\end{align}
When $\alpha=1$ this gives the expansion in terms of Schur polynomials. Comparing with \eqref{njack}, we have
\begin{equation}
\det (I - X)^{-a} = \sum_{\mu}\frac{[a]_\mu^{(1)}}{d'_\mu|_{\alpha=1}} s_{\mu}(X). \label{binschur}
\end{equation}
\begin{theorem} 
\label{theorem_average_as_2F1}
	Consider the averages $\eqref{R_definition}$ with the truncation $A$ replaced by $AV$ where $V$ is an $M \times M$ deterministic matrix of complex numbers. Then for $\beta \in \{2,4\}$ we have
	\begin{align}
		R^{(\beta)}_{\gamma}(x;V)
		=
		|x|^{\gamma M}
		\
		_2F_1^{(\frac{2}{\beta})}\left(-\frac{\gamma}{2}, -\frac{\gamma}{2}+ 1 - \frac{\beta}{2}; \frac{\beta}{2} N; \frac{1}{|x|^{2}} \Sigma\right),
		\label{char_pol_pre_integral}
	\end{align}
	where $\mathrm{Re}(\gamma) > -1$ and $\Sigma = VV^{\dagger}$. If $\beta=1$ then \eqref{char_pol_pre_integral} continues to hold with the replacements $|x|^{2} \to x^{2}$ and $VV^{\dagger} \to VV^{\mathrm{T}}$, while if $\beta=4$ the matrix $V$ has real quaternion entries. 
\end{theorem}
\begin{remark}
If $\gamma$ is not an integer, we also require the bound on the operator norm $|| \frac{1}{|x|^2} \Sigma || < 1$ for the $_2F_1$ series to converge absolutely. If $\gamma =2k$ with $k \in \mathbb{N}$, the case mainly considered here, the series \eqref{char_pol_pre_integral} terminates and consists of only finitely many terms. Then the condition on the operator norm is not required.
\end{remark}
\begin{proof}[Proof of Theorem \ref{theorem_average_as_2F1}]
We begin with the real case $\beta=1$. Multiplying both sides of \eqref{schur_average_pre_sum_beta_1} by $
\frac{[-\gamma]^{(1)}_\mu}{d'_\mu|_{\alpha=1}}$ yields
\begin{align}
	\mathbb{E}
		\left [
			\frac{[-\gamma]^{(1)}_\mu}{d'_{\mu}|_{\alpha=1}}
			s_\mu (VA)
		\right ]
	&=	
	\frac{1}{\left | \nu \right |!}
	\frac{[-\gamma]^{(1)}_{2\nu}}{2^{2\left | \nu \right |} \left [ \frac{1}{2} N \right ]^{(2)}_\nu}
	C_\nu^{(2)}(VV^T)
	\label{eqnorms}
	,
\end{align}
when $\mu = 2\nu$, and $0$ otherwise. Then we insert the identity $[-\gamma]^{(1)}_{2\nu} = 2^{2|\nu|} [-\gamma/2]^{(2)}_{\nu} 	[(-\gamma+1)/2]^{(2)}_{\nu}$ into the right-hand side of \eqref{eqnorms} and sum both sides over all partitions $\mu$. The sum on the left-hand side is evaluated by recalling the generalised binomial expansion \eqref{binschur}, while the sum on the right-hand side coincides with the definition of the $_2F_1$ hypergeometric function defined in \eqref{pFq_definition}. We thus obtain
\begin{align}
	\mathbb{E}
	\left [
		\det ( I_{M} - VA )^\gamma
	\right ]
	= \
	_2F_1^{(2)}(-\gamma/2, (-\gamma+1)/2; N/2; VV^T)
	.	
\end{align}
For $\beta \in \{2,4\}$, the approach is similar so we only highlight the main differences. For $\beta=2$ we multiply both sides of \eqref{schur_average_pre_sum_beta_2} by $\frac{[-\frac{\gamma}{2}]^{(1)}_{\mu}[-\frac{\gamma}{2}]^{(1)}_{\nu}}{d'_{\mu}|_{\alpha=1}d'_{\nu}|_{\alpha=1}}$. Then summing over $\mu$ and $\nu$ we similarly get 
\begin{align}
	\mathbb{E}
	\left [
		|\det ( I_{M} - VA )|^{\gamma}
	\right ]
	= \
	_2F_1^{(1)}(-\gamma/2, -\gamma/2; N; VV^\dag)
	.
\end{align}
For $\beta=4$ we multiply both sides of \eqref{schur_average_pre_sum_beta_4} by $\frac{[-\gamma]^{(1)}_\mu}{d'_\mu|_{\alpha=1}}$ and sum over $\mu$ using the identity $[-\gamma]_{\nu^2}^{(1)} = [-\gamma]_{\nu}^{(1/2)} [-\gamma-1]_{\nu}^{(1/2)}$ to obtain
\begin{align}
	\mathbb{E}
	\left [
		\det ( I_{M} - VA)^\gamma
	\right ]
	= \
	_2F_1^{(1/2)}(-\gamma, -\gamma-1; 2N; VV^\dag)
	.	
\end{align}
Substituting $V \mapsto x^{-1}V $ we obtain \eqref{char_pol_pre_integral} as a consequence of these identities.
\end{proof}

\begin{remark}
The proof of Theorem \ref{theorem_average_as_2F1} makes essential use of the identities of Theorem \ref{theorem_schur_average_pre_sum} which are only stated for $N \geq 2M$. This can be traced back to the result \eqref{ForJPDF} for the eigenvalue distribution of $AA^{\dagger}$ which is singular if $M$ is too large. Despite this, we can express the averages in \eqref{R_definition} in terms of the joint probability density function of the eigenvalues of $A$ provided in the works \cite{Zyczkowski:2000ue,Khoruzhenko_2010} for $\beta \in \{1,2\}$, for $\beta=4$ see \cite{Forrester16} for $N \geq 2M$ and the recent work \cite{KL21} for any $N>M$. For any $\beta \in \{1,2,4\}$, the latter eigenvalue distributions for $A$ are well defined for all $N>M$ with $N$ appearing as a parameter. In this way, Theorem \ref{theorem_average_as_2F1} extends to any $N>M$ by analytic continuation.
\end{remark}

\begin{corollary}
\label{bndrycor}
For any $N\geq M$, $\gamma>0$ and $\theta \in \mathbb{R}$, we have the following identities for $\beta = 1, 2$ and $4$ respectively:
\begin{align}
	&\mathbb{E}
	\left [
		\det ( I_{M} - A )^\gamma
	\right ]
	=
	\prod_{j=N-M+1}^{N}
	\frac{\Gamma(\frac{j}{2} )\Gamma(\frac{j-1}{2} +\gamma)}{\Gamma(\frac{j}{2}+\frac{\gamma}{2}) \Gamma(\frac{j-1}{2}+ \frac{\gamma}{2})}
	,
	\label{truncation_at_edge_beta_1}
	\\
	&
	\mathbb{E}
	\left [
		|\det ( e^{i \theta} I_{M} - A )|^{\gamma}
	\right ]
	=
	\prod_{j=N-M+1}^{N}
	\frac{\Gamma(j) \Gamma(j+ \gamma)}{ \left ( \Gamma(j + \frac{\gamma}{2}) \right )^2}
	,
	\label{truncation_at_edge_beta_2}	
	\\
	&
	\mathbb{E}
	\left [
		\det ( e^{i \theta} I_{M} - A )^\gamma
	\right ]
	= 
	\prod_{j=N-M+1}^{N}
	\frac{\Gamma(2j) \Gamma(2j+ 2\gamma + 1)}{\Gamma(2j+ \gamma) \Gamma(2j + \gamma + 1)}.
	\label{truncation_at_edge_beta_4}
\end{align}
\end{corollary}
\begin{proof}
In the boundary case $V = \Sigma = I_{M}$ with $x=\pm 1$ for $\beta = 1$ and $x = e^{i \theta}$ for $\beta \in \{2, 4\}$, the three evaluations \eqref{truncation_at_edge_beta_1}, \eqref{truncation_at_edge_beta_2} and \eqref{truncation_at_edge_beta_4} follow from \eqref{char_pol_pre_integral} and the generalised Gauss summation \cite[Corollary 3.5]{Y92}
\begin{align}
	_2F_1^{(1/\alpha)}(a,b;c;I_{M})
	=
	\prod_{j=0}^{M-1}
	\frac{\Gamma(c- \alpha j) \Gamma(c-a-b- \alpha j)}{\Gamma(c-a- \alpha j) \Gamma(c-b - \alpha j)}
	.
\end{align}
\end{proof}
\begin{remark}
\label{rem:sorem}
The three identities \eqref{truncation_at_edge_beta_1}, \eqref{truncation_at_edge_beta_2} and \eqref{truncation_at_edge_beta_4} are natural generalisations of moment formulae obtained by Keating and Snaith in the context of the classical compact groups. In our setting this correponds to the specialisation $M=N$. As pointed out in the introduction, when $M=N$ and $\beta=2$, the identity \eqref{truncation_at_edge_beta_2} specialises to the moments of characteristic polynomials in $\mathrm{U(N)}$ obtained in \cite{Keating_Snaith_2000}. Furthermore, when $M=N$ and $\beta=4$, the identity \eqref{truncation_at_edge_beta_4} specialises to the moments of characteristic polynomials in $\mathrm{Sp(2N)}$ obtained in \cite{Keating:2000ub},
\begin{equation}
\mathbb{E}_{\mathrm{Sp(2N)}}\left[\det(e^{i\theta}I_{N}-A)^{\gamma}\right] = 2^{2N\gamma}\prod_{j=1}^{N}\frac{\Gamma(1+N+j)\Gamma(1/2+\gamma+j)}{\Gamma(1/2+j)\Gamma(1+\gamma+N+j)}. \label{sp2nidentity}
\end{equation}
The agreement of \eqref{sp2nidentity} and \eqref{truncation_at_edge_beta_4} with $M=N$ can be obtained with some simple manipulations of the product and the duplication formula for the Gamma function. We stress that in formulae \eqref{truncation_at_edge_beta_4} and \eqref{sp2nidentity} we view $e^{i\theta}$ as a real quaternion, hence when viewed as $2N \times 2N$ complex matrices the argument $e^{i\theta}$ must be real, as in the work \cite{Keating:2000ub}.

When $\beta=1$ the corresponding result in \cite{Keating:2000ub} is for $\mathrm{SO(2N)}$ which is the subgroup of $\mathrm{O(2N)}$ consisting of orthogonal matrices with determinant $1$,
\begin{equation}
\mathbb{E}_{\mathrm{SO(2N)}}\left[ \det ( I_{2N} - A )^\gamma\right] = 2^{2N\gamma}
	\prod_{j=1}^{N}
	\frac{\Gamma(N - 1 + j )\Gamma(\gamma - \frac{1}{2} + j)}{\Gamma(-\frac{1}{2}+j) \Gamma(N - 1 + \gamma + j)} \label{KS}.
	\end{equation}
To see the connection with our result, we note the disjoint union $\mathrm{O(2N)} = \mathrm{SO(2N)}\cup \mathrm{O^{-}(2N)}$, where $\mathrm{O^{-}(2N)}$ consists of the elements of $\mathrm{O(2N)}$ that have determinant $-1$. Matrices in $\mathrm{O^{-}(2N)}$ have trivial eigenvalues $1$ and $-1$, so they do not contribute to the moments of $\det(I_{M}-A)$. By relating the Haar measure on $\mathrm{O(2N)}$ and $\mathrm{SO(2N)}$ we find for $\gamma>0$,
\begin{equation}
\mathbb{E}_{\mathrm{O(2N)}}\left[ \det ( I_{2N} - A )^\gamma\right] = \frac{1}{2}\,\mathbb{E}_{\mathrm{SO(2N)}}\left[ \det ( I_{2N} - A )^\gamma\right]. \label{OSO}
\end{equation}
It is worth cautioning that \eqref{OSO} is not valid when $\gamma=0$ because of an additional $\frac{1}{2}$ contributing from $\mathrm{O^{-}(2N)}$. After inserting the $M=N$ specialisation of \eqref{truncation_at_edge_beta_1} into \eqref{OSO} we reclaim the Keating and Snaith result \eqref{KS} by similar manipulations and duplication identities used in the symplectic case. 
\end{remark}
\begin{theorem}
\label{propv}
For $\beta \in \{2,4\}$ we have
	\begin{equation}
R^{(\beta)}_{2k}(x;V) = \frac{\det(|x|^{2}-\Sigma)^{k}}{S_{k,N}(\beta)}
	\int_{[0,1]^k} \prod_{i=1}^{k}dt_{i}\, t_i^{N-M}
	\prod_{j=1}^{M} (t_i - \gamma_j)
	 |\Delta(\vec{t})|^{\beta'}
		\label{char_pol_integral}
		,
	\end{equation}
	where $\Sigma = VV^{\dagger}$, $\{\gamma_{j}\}_{j=1}^{M}$ are the eigenvalues of $-\Sigma(|x|^{2}-\Sigma)^{-1}$,  and $S_{k,N}(\beta)$ is given by \eqref{skn}. If $\beta=1$, \eqref{char_pol_integral} continues to hold with the replacements $|x|^{2} \to x^{2}$ and $VV^{\dagger} \to VV^{\mathrm{T}}$.
\end{theorem}
\begin{proof}[Proof of Theorems \ref{propv} and \ref{theorem_truncated_duality_integrals_sigma}]
Our starting point is the hypergeometric evaluation \eqref{char_pol_pre_integral} with $\gamma=2k$ and $k \in \mathbb{N}$. The following identity is a consequence of the so-called \textit{Kummer relations} for hypergeometric functions of matrix argument, see \textit{e.g.} \cite[Theorem 7.4.3]{Muirhead_Aspects} for the real case or \cite[Theorem 3.2]{Y92} for arbitrary $\beta>0$,
\begin{align}
&|x|^{2kM}\,_2F_1^{(\frac{2}{\beta})}\left(-k, -k+ 1 - \frac{\beta}{2}; \frac{\beta}{2} N; \frac{1}{|x|^{2}} \Sigma\right) \\
&= \det(|x|^{2}I_{M}-\Sigma)^{k}\,_2F_1^{(\frac{2}{\beta})}\left(-k, \frac{\beta}{2}N+k- 1 + \frac{\beta}{2}; \frac{\beta}{2} N; X\right), \label{pfaffid}
\end{align}
where $X = -\Sigma(|x|^{2}I_{M}-\Sigma)^{-1}$. To complete the proof we make use of the following integral  representation of $_2F_1^{(\alpha)}$ due to Kaneko \cite[Theorem 5]{Kaneko_1993},
\begin{align}
	&_2F_1^{(\alpha)}\left(-k,\frac{1}{\alpha}(\lambda_1 + \lambda_2 + M + 1) + k - 1;\frac{1}{\alpha}(\lambda_1 + M );X\right)
	\nonumber
	\\
	&=
	\frac{1}{S_{k}(\lambda_1+M,\lambda_2,2\alpha)}
	\int_{[0,1]^{k}} \prod_{i=1}^{k}dt_{i}\,t_i^{\lambda_1}
	(1 - t_i)^{\lambda_2}\prod_{j=1}^{M}(t_i -\gamma_j) |\Delta(\vec{t})|^{2\alpha}
	,
	\label{Kaneko_2F1_integral_representation}
\end{align}
where $S_{k}(\lambda_1+M,\lambda_2,2\alpha)$ is given by \eqref{selberg} and $\{\gamma_{j}\}_{j=1}^{M}$ are the eigenvalues of $X$. Applying \eqref{Kaneko_2F1_integral_representation} to \eqref{pfaffid} with parameter choices $\alpha = \frac{2}{\beta}$, $\lambda_{1} = N-M$ and $\lambda_{2}=0$ results in the claimed identity \eqref{char_pol_integral}. Setting $V = \Sigma = I_{M}$ in \eqref{char_pol_integral} results in \eqref{truncation_characteristic_poly}.
\end{proof}
\begin{remark}
In the case $\beta=1$, a duality also occurs for the odd moments of the characteristic polynomial, see \cite{Forrester_2009} for a version of this in the real Ginibre ensemble. Setting $\gamma = 2k+1$ with $k \in \mathbb{N}$ in \eqref{char_pol_pre_integral}, we see that the second parameter of the hypergeometric function is a negative integer. Applying a similar procedure using \eqref{pfaffid} and \eqref{Kaneko_2F1_integral_representation} we obtain
	\begin{align}
\mathbb{E}\left [\det ( xI - AV )^{2k + 1}
		\right ] = \frac{1}{S_{k}(N,2,4)}\int_{[0,1]^{k}}\prod_{i=1}^{k}dt_{i}\,t_i^{N-M}(1 - t_i)^{2}\prod_{j=1}^{M}(t_{i}-\gamma_{j})|\Delta(\vec{t})|^{4}. \label{char_pol_real_integral_odd}
	\end{align}

\end{remark}

\section{Asymptotics: Proofs of Theorems \ref{theorem_strong_non-unitarity_uniform_asymptotics} and \ref{theorem_dulaity_integrals_asymptotics}}
\label{sec:asympt}
In this section, based on the duality formulae in Theorem \ref{theorem_truncated_duality_integrals_sigma} we shall obtain the asymptotic behaviour of the averages $R^{(\beta)}_{2k}(x)$. We begin with the strong non-unitarity regime.
\begin{proof}[Proof of Theorem \ref{theorem_strong_non-unitarity_uniform_asymptotics}]
We shall apply the Laplace method of asymptotics to the integral representation \eqref{truncation_characteristic_poly}. Throughout the proof we denote $\mu = \frac{M}{N} \equiv \frac{M}{N_{M}}$ so that $\mu \to \tilde{\mu} \in (0,1)$ as $M \to \infty$. We start by writing \eqref{truncation_characteristic_poly} in the form
\begin{align} 
R^{(\beta)}_{2k}(x) = \frac{1}{S_{k}(M/\mu, 0, \beta')}
	\int_{[0,1]^k}
	\prod^{k}_{j = 1}dt_{j}\,
	e^{M\varphi(t_{j})}|\Delta(\vec{t})|^{\beta'} \label{im},
\end{align}
where the action is given by
\begin{align}
	\varphi(t) = \log (1 + (|x|^2 - 1)t) + (\mu^{-1} - 1) \log t.
\end{align}
The function $\varphi(t)$ attains a unique maximum at the point $\displaystyle{t^*= \frac{1-\mu}{1-|x|^2}}$. Expanding $\varphi$ near $t=t^{*}$ yields
\begin{align}
	\varphi(t)
	=
	\log \mu + (\mu^{-1}-1)\log(t^{*}) 
	-
	\frac{(1-|x|^{2})^2}{2(1 - \mu)\mu^2}
	( t - t^{*})^2
	+
	O((t-t^*)^3)
	.
	\label{strong_non-unitarity_exponent_expansion}
\end{align}
Note that it is possible for the critical point $t^{*}$ to approach the boundary of the integration domain at $t_{j}=1$, this may happen if $x$ comes close to the circle of radius $\sqrt{\mu}$. In order to obtain a uniform expansion in $x$, we will begin by assuming that $|x|$ is sufficiently close to $\sqrt{\mu}$, namely that $t^{*} > 1-\epsilon$ for some small $\epsilon>0$. For such $t_{*}$ (or such $x$), the main contribution to the integral \eqref{im} comes from the set $[t^{*}-\epsilon,1]^{k}$, with the complementary region $[0,1]^{k} \setminus [t^{*}-\epsilon,1]^{k}$ giving an exponentially suppressed contribution. On the set $[t^{*}-\epsilon,1]^{k}$ we apply \eqref{strong_non-unitarity_exponent_expansion} resulting in the uniform approximation
\begin{align} 
&R^{(\beta)}_{2k}(x)
	\sim
	\frac{\mu^{Mk}(t^{*})^{M k(\mu^{-1} - 1)}}{S_{k}(M/\mu, 0, \beta')}
	\int_{[t^{*}-\epsilon,1]^{k}}
	\prod^{k}_{j = 1}dt_{j}\,e^{-M\frac{(1-|x|^{2})^2}{2(1 - \mu)\mu^2}(t_{j}-t^{*})^2}|\Delta(\vec{t})|^{\beta'}\\
	&\sim
	C_{k,\beta}(\mu,x)M^{\frac{k^{2}}{\beta}+\frac{k}{2}(1-\frac{2}{\beta})}\,\mu^{Mk}(t^{*})^{M k(\mu^{-1} - 1)}
	\int_{(-\infty, \sqrt{M} g(\mu,x) )^k}
	\prod_{j=1}^{k}dt_{j}\,
	e^{
	-\frac{1}{2}t_j^2
	}
	|\Delta(\vec{t})|^{\beta'},
	\label{strong_non-unitarity_asymptotics_uniform_integral}
\end{align}
where $g(\mu,x) = \frac{\mu - |x|^2}{\mu \sqrt{1 - \mu}}$ and
\begin{equation}
C_{k,\beta}(\mu,x) = \left ( \frac{\sqrt{1 - \mu}}{1-|x|^2} \right )^{k + \frac{2}{\beta}k(k-1)}\,\prod_{j=0}^{k-1}\frac{\Gamma(1+\frac{2}{\beta})}{\Gamma(1+j\frac{2}{\beta})\Gamma(1+(1+j)\frac{2}{\beta})}. \label{ckmu}
\end{equation}
To obtain the final estimate in \eqref{strong_non-unitarity_asymptotics_uniform_integral} we changed variables and inserted the asymptotics of $S_{k}(M/\mu, 0, \beta')$ using \eqref{selberg} and Stirling's formula. Then we recognise the final integral in \eqref{strong_non-unitarity_asymptotics_uniform_integral} as the distribution function of the largest eigenvalue of a $k \times k$ random matrix from the Gaussian $\beta'$-Ensemble, up to a normalization constant $G_{\beta,k}$, 
\begin{align}
	\int_{(- \infty, \sqrt{M} g(\mu,x))^{k}}
	\prod^{k}_{j = 1}dt_{j}\,
	e^{
	-\frac{1}{2}t_j^2
	}
	|\Delta(\vec{t})|^{\beta'} =
	G_{\beta,k}\,
	\mathbb{P} (\lambda_{\mathrm{max}}^{(\mathrm{G\beta'E})} < \sqrt{M} g(\mu,x)) \label{gbetaemax}
	,
\end{align}
where
\begin{align}
	G_{\beta,k}	
	=
	(2\pi)^{\frac{k}{2}}\prod_{j=0}^{k-1}
	\frac{\Gamma (1 + (1+j)\frac{2}{\beta})}{\Gamma (1 + \frac{2}{\beta})}
	.
\end{align}
Inserting \eqref{gbetaemax}, \eqref{ckmu} and the explicit formulae for $t^{*}$ into \eqref{strong_non-unitarity_asymptotics_uniform_integral} completes the proof provided $t^{*}>1-\epsilon$. If $t^{*} < 1-\epsilon$ we have that $x$ stays a fixed distance away from the boundary $|x| = \sqrt{\mu}$, in which case the $\mathrm{G\beta' E}$ probability can be replaced with $1$. Then Theorem \ref{theorem_strong_non-unitarity_uniform_asymptotics} follows by repeating the same steps, now applying the standard (two-sided) Laplace method with a critical point strictly within the interior of the integration domain.
\end{proof}

Next we consider the \textit{weak non-unitarity} regime, where we keep $\kappa = N - M$ fixed and let $M \to \infty$.
\begin{proof}[Proof of Theorem \ref{theorem_dulaity_integrals_asymptotics}]
We set $|x|^{2}=1-\frac{2u}{M}$ in the expression \eqref{truncation_characteristic_poly}. Making use of the limit $\left(1-\frac{2ut_{i}}{M}\right)^{M} \to e^{-2ut_{i}}$ as $M \to \infty$ and changing variables $t_{i} \to \frac{t_{i}}{2u}$ for each $i=1,\ldots,k$ we get
\begin{align}
R^{(\beta)}_{2k}(x) \sim
	\frac{1}{S_{k}(M+\kappa,0,\beta')}
	(2u)^{-\frac{2k^{2}}{\beta} - k(\kappa + 1 - \frac{2}{\beta} )}
	\int_{[0,2u]^{k}}\prod_{i=1}^{k}dt_{i}\,
	e^{- t_i }
	t_i^{\kappa}|\Delta(\vec{t})|^{\beta'}
	, \label{lbetaeint}
\end{align}
as $M \to \infty$. The asymptotics of $S_{k}(M+\kappa,0,\beta')$ follows from \eqref{selberg} and Stirling's formula,
\begin{equation}
S_{k}(M+\kappa,0,\beta') \sim M^{-\frac{2k^{2}}{\beta}-\left(1-\frac{2}{\beta}\right)k}\prod_{j=0}^{k-1}\frac{\Gamma(1+j\frac{2}{\beta})\Gamma(1+(j+1)\frac{2}{\beta})}{\Gamma(1+\frac{2}{\beta})}. \label{sasy}
\end{equation}
Then we recognise the integral in \eqref{lbetaeint} as the largest eigenvalue distribution of the Laguerre $\beta'$-Ensemble,
\begin{equation}
\int_{[0,2u]^{k}}\prod_{i=1}^{k}dt_{i}\,
	e^{- t_i }
	t_i^{\kappa}|\Delta(\vec{t})|^{\beta'} = W_{k,\kappa,\beta}\,\mathbb{P}(\lambda_{\mathrm{max}}^{(\mathrm{L\beta'E})} < 2u) \label{lmax}
\end{equation}
where the normalization constant is given by 
\begin{equation}
W_{k,\kappa,\beta} = \prod_{j=0}^{k-1}\frac{\Gamma(1+(j+1)\frac{2}{\beta})\Gamma(\kappa+1+j\frac{2}{\beta})}{\Gamma(1+\frac{2}{\beta})},
\end{equation}
see e.g. \cite[Theorem 4.7.3]{Forrester_2010_Log}. Inserting \eqref{lmax} and \eqref{sasy} into \eqref{lbetaeint} completes the proof of the Theorem.
\end{proof}
\section{The distribution of the log-modulus of the characteristic polynomial} 
\label{section_limit_theorems}
In this section we consider the limiting fluctuations of the random variables $\log|\det(A)|$ and also for the log-modulus of the characteristic polynomial evaluated on the unit circle.
\begin{lemma}
\label{lem:momgenfn}
The moment generating function of $\log|\det A|$ is
\begin{align}
	\mathbb{E}
	\left [
	| \det A |^\gamma
	\right ]
	&=
	\prod^{M-1}_{j=0}
	\frac{
	\Gamma( \frac{\gamma}{2} + \frac{\beta}{2} + \frac{\beta}{2} j)
	\Gamma( \frac{\beta}{2} N - \frac{\beta}{2} M + \frac{\beta}{2} + \frac{\beta}{2} j )
	}{
	\Gamma( \frac{\beta}{2} + \frac{\beta}{2} j) 
	\Gamma( \frac{\beta}{2} N - \frac{\beta}{2} M + \frac{\beta}{2} + \frac{\gamma}{2} + \frac{\beta}{2} j)
	}
	\label{moments_det_beta_truncation_finite_size_N}.
\end{align}
\end{lemma}
\begin{proof}
We use that $| \det A |^\gamma = ( \det AA^\dag )^{\frac{\gamma}{2}}$ and apply the explicit joint distribution of eigenvalues of $AA^{\dagger}$ given in \eqref{ForJPDF}. This yields
\begin{align}
	\mathbb{E}
	\left [
	| \det A |^\gamma
	\right ]
	&=	
	\frac{1}{S_{M}\left(a,b,\lambda\right)}
	\int_{[0,1]^M}
	\prod^M_{j=1}dy_{j}\,
	y_j^{\frac{\gamma}{2} + a - 1 }
	(1 - y_j)^{b-1}|\Delta(\vec{y})|^{\beta}\\
	&= \frac{S_{M}\left(a+\frac{\gamma}{2},b,\lambda\right)}{S_{M}\left(a,b,\lambda\right)}, \label{selbergratio}
\end{align}
where the parameters are $a = \frac{\beta}{2}$, $b = \frac{\beta}{2}(N-2M+1)$ and $\lambda = \frac{\beta}{2}$. Inserting the explicit evaluation \eqref{selberg} into \eqref{selbergratio} and simplifying completes the proof of the Lemma.
\end{proof}
Next recall that a \textit{Beta random variable with parameters $a$ and $b$}, denoted $\mathcal{B}_{a,b}$, is the random variable taking values on $(0,1)$ with probability density function
\begin{equation}
p_{a,b}(x) = \frac{1}{B(a,b)}x^{a-1}(1-x)^{b-1}, \qquad  B(a,b) = \frac{\Gamma(a)\Gamma(b)}{\Gamma(a+b)}. \label{betapdf}
\end{equation}
Note that general non-integer moments of $\mathcal{B}_{a,b}$ are explicitly available,
\begin{equation}
\mathbb{E}(\mathcal{B}_{a,b}^{\frac{\gamma}{2}}) = \frac{\Gamma(a+\frac{\gamma}{2})\Gamma(a+b)}{\Gamma(a+b+\frac{\gamma}{2})\Gamma(a)}. \label{betamgf}
\end{equation}
\begin{lemma}
\label{lem:eqdist}
We have the following two equalities in distribution,
\begin{equation}
\log|\det(A)| \overset{d}{=} \frac{1}{2}\sum_{j=0}^{M-1}\log \mathcal{B}_{\alpha_j,\frac{\beta}{2}\kappa}, \label{eqdist1}
\end{equation}
and
\begin{align}
		\log | \det (A)| \overset{d}{=} \frac{1}{2}\sum_{j=0}^{\kappa-1}\log \mathcal{B}_{\alpha_{j},\frac{\beta}{2}M}, \label{eqdist2}
	\end{align}
	where $\alpha_{j} = \frac{\beta}{2}(1+j)$, $\kappa = N-M$ and all Beta random variables are independent.
\end{lemma}
\begin{proof}
Both equalities are obtained by calculating the moment generating functions on both sides of \eqref{eqdist1} and \eqref{eqdist2} using identities \eqref{moments_det_beta_truncation_finite_size_N} and \eqref{betamgf}. This immediately yields the first equality \eqref{eqdist1}. For the second equality \eqref{eqdist2} one needs the following rearrangement of \eqref{moments_det_beta_truncation_finite_size_N},
\begin{equation}
\prod^{M-1}_{j=0}
	\frac{
	\Gamma( \frac{\gamma}{2} + \frac{\beta}{2} + \frac{\beta}{2} j)
	\Gamma( \frac{\beta}{2}\kappa+ \frac{\beta}{2} + \frac{\beta}{2} j )
	}{
	\Gamma( \frac{\beta}{2} + \frac{\beta}{2} j) 
	\Gamma( \frac{\beta}{2} \kappa + \frac{\beta}{2} + \frac{\gamma}{2} + \frac{\beta}{2} j)
	}
=
	\prod^{\kappa-1}_{j=0}
		\frac{
		\Gamma(\frac{\beta}{2}(M + 1 + j)) \Gamma(\frac{\gamma}{2} +\frac{\beta}{2}(1 + j))
		}{
		\Gamma(\frac{\beta}{2}(1 + j)) \Gamma(\frac{\gamma}{2} +\frac{\beta}{2}(M+ 1 + j))
		}. \label{rearrange}
\end{equation}
\end{proof}
\begin{remark}
When $\beta \in \{2,4\}$, the first equality \eqref{eqdist1} is in agreement with recent results for the spectral radii of random truncations \cite{Dubach:2018_powers_ginibre,Dubach:2021_symmetries_GinSE}.
\end{remark}
The equality \eqref{eqdist2} is particularly helpful for extracting the limiting behaviour in the regime of weak non-unitarity where the upper limit of the summation $\kappa = N-M$ is kept fixed.
\begin{proof}[Proof of Theorem \ref{theorem_log_determiants_limit_theorem}]
This follows immediately from Lemma \ref{lem:eqdist} and the fact that the renormalised independent random variables $\frac{\beta}{2}M\mathcal{B}_{\alpha_{j},\frac{\beta}{2}M}$ converge in distribution to independent Gamma random variables $\Gamma_{j}^{(\beta)}$ as $M \to \infty$, also in the sense of moment generating functions. Thus we have
\begin{align}
\mathbb{E}
	\left [
		e^{
		\frac{\gamma}{2}
		\sum^{\kappa-1}_{j=0}
		\log \mathcal{B}_{\alpha_{j},\frac{\beta}{2}M}
		}
	\right ]&=
	e^{-\frac{\gamma}{2}\kappa \log \frac{\beta}{2} M}
	\mathbb{E}
	\left [
		e^{
		\frac{\gamma}{2}
		\sum^{\kappa-1}_{j=0}
		\log 
		\left (
		\frac{\beta}{2} M
		\mathcal{B}_{\alpha_{j},\frac{\beta}{2}M}
		\right )
		}
	\right ]
	\nonumber
	\\	
	&\sim
	\left ( \frac{\beta}{2} M \right )^{-\frac{\gamma}{2}\kappa}
	\mathbb{E}
	\left [
		e^{
		\frac{\gamma}{2}
		\sum^{\kappa-1}_{j=0}
		\log 
		\Gamma^{(\beta)}_j
		}
	\right ].
\end{align}
\end{proof}

We remark that in some sense the exact result of Lemma \ref{lem:momgenfn} may be regarded as a particular case of earlier results that appeared for Hermitian ensembles. This is because of the relation $|\det A |^{\gamma} = (\det AA^{\dagger})^{\frac{\gamma}{2}}$ which effectively Hermitises the problem. The determinant of the matrices $AA^{\dagger}$ was considered in the work \cite{R07} in relation to Jacobi ensembles, where the rearrangement \eqref{rearrange} is also noted. In contrast to Theorem \ref{theorem_log_determiants_limit_theorem}, the work \cite{R07} considers a regime equivalent to what we refer to as strong non-unitarity.
\begin{proposition}
 \label{theorem_logdet_snu_limit}
Let $A$ be a real ($\beta = 1$), complex ($\beta = 2$) or real quaternion ($\beta = 4$) $M \times M$ truncation of a Haar distributed matrix taken from $\mathrm{O(N)}, \mathrm{U(N)}, \mathrm{Sp(2N)}$. Then for $\mu := \frac{M}{N} \to \tilde{\mu} $ as $M \to \infty$, we have the convergence in distribution to a standard normal random variable
	\begin{align}
		\frac{\log | \det(A) |  - \kappa_1^{(\beta)}(M)}{\sqrt{\frac{1}{2\beta}\log(M)}}	
		\overset{d}{\underset{M\to \infty}{\longrightarrow}}
		\mathcal{N}(0,1)
		,
	\end{align}
	where $\kappa_1^{(\beta)}(M)$ is the mean of $\log|\det(A)|$ given by
	\begin{align}
		\kappa_1^{(\beta)}(M)
		=
		\frac{1}{2} M \log \left(\mu ( 1 - \mu )^{(\mu^{-1} - 1)}\right)
		+ 
		\frac{1}{4}\,\left(\frac{2}{\beta}-1\right) \log M + O(1), \qquad M \to \infty.
	\end{align}
\end{proposition}
\begin{proof}
See \cite[Theorem 3.8]{R07}.
\end{proof}
We point out that the structure of the moments \eqref{moments_det_beta_truncation_finite_size_N} as well as those appearing in Corollary \ref{bndrycor} are all particular products of Gamma functions. In the special case $\beta=2$ and $N=M$, i.e. for characteristic polynomials in $\mathrm{U(N)}$, the probabilistic content of such evaluations in terms of Beta random variables has been pointed out \cite{BHNY08}. Obtaining the asymptotics of such products is a by now fairly routine, though sometimes rather cumbersome, application of so-called Barnes G-functions and their asymptotic behaviour \cite{Keating:2000ub, Keating_Snaith_2000,R07,Chen:2017uw}. Fortunately, as we shall see below, we are able to avoid any direct computations of this nature by relating our product forms to those considered in other contexts.

For characteristic polynomials of truncations evaluated on the unit circle it is less obvious how to reduce the problem to one of the Hermitian ensembles discussed above. However, as a consequence of our hypergeometric evaluations in Corollary \ref{bndrycor}, the moments are expressible in terms of averages over the classical compact groups. In fact this can be done not just for $N=M$ but for all $N \geq M$. 
\begin{lemma}
\label{haardecomp}
Let $A$ be the truncation in \eqref{uconstruct} and let $\theta \in \mathbb{R}$. Let $U_{N}$ and $V_{N-M}$ be two independent random matrices chosen with respect to Haar measure from the corresponding classical compact groups of order $N$ and $N-M$ respectively. Then for all $N \geq M$ and all three groups corresponding to $\beta \in \{1,2,4\}$, we have the equality in distribution
\begin{equation}
\log|\det(e^{i\theta}I_{M}-A)| \overset{d}{=} \log|\det(e^{i\theta}I_{N}-U_{N})|-\log|\det(e^{i\theta}I_{N-M}-V_{N-M})|, \label{eqdist}
\end{equation}
where $\theta \in \{0,\pi\}$ if $\beta=1$ and if $N=M$ we omit the second term in \eqref{eqdist}.
\end{lemma}
\begin{proof}
This follows from calculating the moment generating function on both sides of \eqref{eqdist} using Corollary \ref{bndrycor} and the observations in Remark \ref{rem:sorem}.
\end{proof}
Lemma \ref{haardecomp} connects characteristic polynomials of $A$ to those of the classical compact groups. Then we can exploit known results for the latter matrices to complete the proof of Theorem \ref{limit_theorem_logdet_at_I}.
\begin{proof}[Proof of Theorem \ref{limit_theorem_logdet_at_I}]
We shall use Lemma \ref{haardecomp} as our starting point. Let $\phi_{N}(\gamma)$ denote the moment generating function of the random variable $\log|\det(e^{i\theta}I_{N}-U_{N})|$ where $U_{N}$ is a Haar distributed random matrix of size $N \times N$ from either $\mathrm{O(N)}, \mathrm{U(N)}$ or $\mathrm{Sp(2N)}$. The asymptotics of $\phi_{N}(\gamma)$ were obtained in precise detail by Keating and Snaith \cite{Keating:2000ub, Keating_Snaith_2000},
\begin{equation}
\log\phi_{N}(\gamma) = \gamma \,e_{\beta}\log(N) + \frac{\gamma^{2}}{2}v_{\beta}\log(N) + \mathcal{E}_{N}(\gamma), \label{ksexpan}
\end{equation}
where $\mathcal{E}_{N}(\gamma) \to \mathcal{E}(\gamma)$ as $N \to \infty$ and $\mathcal{E}(\gamma)$ is given explicitly in terms of the Barnes G-function. In the regime of weak non-unitarity the contribution from $\phi_{N-M}(\gamma)$ stays of order $1$ and does not contribute to the limiting fluctuations. Hence \eqref{N01} follows from \eqref{ksexpan} and Lemma \ref{haardecomp}. In the regime of strong non-unitarity we have $N-M = N(\mu-1)$ and both terms in Lemma \ref{haardecomp} contribute. Again applying \eqref{ksexpan}, we have
\begin{equation}
\lim_{N \to \infty}\frac{\phi_{N}(\gamma)}{\phi_{N-M}(\gamma)} = \lim_{N \to \infty}\frac{\phi_{N}(\gamma)}{\phi_{N(1-\mu)}(\gamma)} =e^{\gamma\,m_{\tilde{\mu}}+\frac{\gamma^{2}}{2}\,\sigma^{2}_{\tilde{\mu}}},
\end{equation}
and \eqref{Nmv} follows.
\end{proof}

\section{Appendix: direct proof of Corollary \ref{bndrycor} via dimension iterations}
\vspace{5pt}
\hfill \textit{By Guillaume Dubach}. \medskip

In this Appendix, we give an alternative proof of Theorem \ref{prop:introbndry} inspired by the method of \cite{BHNY08}, and indicate how the same method also applies to the cases $\beta= 1$ and $4$ to provide a complete proof of Corollary \ref{bndrycor}. \medskip

The following elementary remarks play an important role: if $A$ is the $M \times M$ truncation of an $N \times N$ matrix $U$ as in \eqref{uconstruct}, then $A=TUT^*$ where $T \in \mathcal{M}_{MN}(\mathbb{C})$ is defined by\footnote{In what follows, we use double lines $\|$ to indicate columnwise separation.}
\begin{equation}
T := \left( I_M \ \big\| \ 0_{M \times (N-M)} \right).
\end{equation}
We also recall Sylvester's identity: if $C \in \mathcal{M}_{N,M}(\mathbb{C}), D \in \mathcal{M}_{M,N} (\mathbb{C})$, then
\begin{equation}\label{Sylvester}
\det(I_N - CD) = \det (I_M - DC),
\end{equation}
a well-known fact whose proof is for instance given in \cite{DeiftGioev}. Finally, we will denote 
\begin{equation}
I_N^{(M)} := 
\left(
\begin{array}{cc}
I_M & 0 \\
0 & 0_{N-M}
\end{array}
\right)
= T^* T
\end{equation}
and, if $v$ is a vector, $v^{(M)}:= I_N^{(M)} v$ , that is, the same vector with the last $N-M$ coefficients cancelled.

\begin{theorem} 
If $A$ is an $M \times M$ truncation of a unitary matrix distributed according to the Haar measure on $\mathrm{U(N)}$ with $N \geq M$, then
$$
\det\left( I_M - A \right)
\stackrel{d}{=}
\prod_{k=1}^M \left( 1- e^{i \omega_k} \sqrt{\beta_{1,N-k}} \right),
$$
where all variables involved are independent, and $(\omega_k)_{k=1}^{M}$ are uniformly distributed on $[0,2\pi]$.
\end{theorem}

Theorem \ref{prop:introbndry} follows from this equality in distribution by a direct computation, which is performed in detail in \cite{BHNY08}.

\begin{proof} We proceed by induction on $M$; the result holds for $M=1$ and any $N$, because we know that a size $1$ marginal of the uniform distribution on $\mathcal{S}_{\mathbb{C}}^{N-1}$ is distributed like 
\begin{equation}
\frac{Z_1 }{\sqrt{ |Z_1|^2 + \cdots + |Z_N|^2 }}
\stackrel{d}{=} e^{i \omega} \sqrt{\beta_{1,N-1}}.
\end{equation}
where $Z_1, \dots, Z_N$ are i.i.d. standard complex Gaussian variables. For general $M \leq N$, we first write $A = T V_N T^*$, with $V_N$ being Haar-distributed on $\mathrm{U(N)}$, and use Sylvester's identity \eqref{Sylvester} to transform the determinant:
\begin{equation}
\det (I_M - A)
= \det\left( I_M - T V_N T^* \right)
= \det\left( I_N - T^* T V_N \right) 
= \det\left( I_N - I_N^{(M)} V_N \right).
\end{equation}
We then use the recursive construction of the Haar measure presented in \cite{BHNY08}, i.e. the fact that 
\begin{equation}\label{rec_distr_CUE}
V_N
\stackrel{d}{=} 
Q
\left(
\begin{array}{cc}
1 & 0 \\
0 & V_{N-1}
\end{array}
\right)
\end{equation}
where $V_{N-1}$ is Haar-distributed on $\mathrm{U(N-1)}$, and $Q$ is a unitary matrix independent of $V_{N-1}$ such that its first column $Q_1$ is uniformly distributed on the sphere -- no other condition is required of $Q$, so we are free to choose it in a suitable way. Namely: we first sample $Q_1$ uniformly on the sphere, and then choose $Q$ to be a reflection with respect to the hyperplane between $Q_1$ and $e_1$. We denote $v := Q_1 - e_1$, a vector orthogonal to the stable hyperplane, so that there are scalars $\lambda_2, \dots \lambda_N$ such that for any $l$,
\begin{equation}
Q_l = e_l + \lambda_l v.
\end{equation}
We then write
\begin{align*}
\det\left( I_N - I_N^{(M)} V_N \right)
& \stackrel{d}{=}
\det\left( I_N -
I_N^{(M)} Q
\left(
\begin{array}{cc}
1 & 0 \\
0 & V_{N-1}
\end{array}
\right)
 \right) \\
& = \det (V_{N-1}) \times
\det \left(
\left(
\begin{array}{cc}
1 & 0 \\
0 & V_{N-1}^*
\end{array}
\right)
- 
I_N^{(M)} Q
\right).
\end{align*}
In order to simplify this last determinant, note that the columns of $I_N^{(M)} Q$ are given by
\begin{equation}
I_N^{(M)} Q_l = I_N^{(M)} (e_l + \lambda_l v) = e_l^{(M)}  + \lambda_l v^{(M)},
\end{equation}
and denote the columns of $V_{N-1}^*$ by $w_1, \dots, w_{N-1}$. We then perform the following elementary operations on the columns: $C_k \mapsto C_k - \lambda_k C_1$, for $k=2, \dots, N$.
\begin{align*}
\det \left(
\left( \hspace{-.05in}
\begin{array}{cc}
1 & 0 \\
0 & V_{N-1}^*
\end{array}
\hspace{-.05in}
\right)
- 
I_N^{(M)} Q
\right) 
& =
\det \left( 
- v^{(M)} \ \Big\| \ 
\left( \hspace{-.05in}
\begin{array}{c}
0 \\
w_1
\end{array} \hspace{-.05in}
\right)
-
e_2^{(M)} - \lambda_2 v^{(M)} \ \Big\|  \cdots
 \Big\| \ \left( \hspace{-.05in}
\begin{array}{c}
0 \\
w_{N-1}
\end{array} \hspace{-.05in}
\right) - e_N^{(M)} - \lambda_N v^{(M)}
\right) \\
& =
\det \left( 
- v^{(M)} \ \Big\| \ 
\left( \hspace{-.05in}
\begin{array}{c}
0 \\
w_1
\end{array} \hspace{-.05in}
\right)
-
 e_2^{(M)} \ \Big\|  \cdots
 \Big\| \ \left( \hspace{-.05in}
\begin{array}{c}
0 \\
w_{N-1}
\end{array} \hspace{-.05in}
\right) - e_N^{(M)}
\right) \\
&= \det
\left(
\begin{array}{cc}
- v_1 & 0 \\
* & V_{N-1}^*- I_{N-1}^{(M-1)}
\end{array}
\right) \\
&= 
-v_1 \det \left( V_{N-1}^*- I_{N-1}^{(M-1)} \right)
\end{align*}
We conclude by applying Sylvester's identity \eqref{Sylvester} to recover a (smaller) truncated matrix:
\begin{equation}\label{last_step}
\det (I_M - A) =
-v_1 \det \left( I_{N-1} - I_{N-1}^{(M-1)} V_{N-1} \right)
= 
-v_1 \det \left( I_{M-1} - \widetilde{T} V_{N-1} \widetilde{T}^*  \right)
\end{equation}
where $\widetilde{T}$ is the $(M-1) \times (N-1)$ version of $T$. By construction, we know that $v_1$ is independent of $V_{N-1}$, with
\begin{equation}
- v_1
= 1 - Q_{1,1}
\stackrel{d}{=} 1 - e^{i \omega} \sqrt{\beta_{1,N-1}}
\end{equation}
where $\omega$ is a uniform argument independent from the beta distribution. The smaller determinant obtained in \eqref{last_step} is similar to the initial term, with $M$ and $N$ both decreased by one; so that the result follows by induction.
\end{proof}

We now briefly explain how the same argument can be adapted to provide a complete proof of Corollary \ref{bndrycor}.

\subsubsection*{Truncated real orthogonal case} The above applies to the truncated real orthogonal case with almost no change, as the decomposition \eqref{rec_distr_CUE} still holds with $V_N, Q, V_{N-1}$ being real orthogonal matrices, and $Q_1$ being uniformly distributed on the real sphere $\mathcal{S}_\mathbb{R}^{N-1}$, so that
\begin{equation}
Q_{1,1}
\stackrel{d}{=}
\frac{X_1 }{\sqrt{ X_1^2 + \cdots + X_N^2 }}
\stackrel{d}{=}
\epsilon \sqrt{\beta_{\frac12,\frac{N-1}{2}}}
\end{equation}
where $X_1, \dots, X_N$ are i.i.d. standard real Gaussian variables, and $\epsilon$ is a random sign such that $\mathbb{P} (\epsilon = 1) = \mathbb{P} (\epsilon = -1) = \frac12$, independent of the beta variable. We chose $Q$ to be a real reflection, and the computation goes unchanged; the conclusion is that, for $\beta=1$,
\begin{equation}
\det\left( I_M - A \right)
\stackrel{d}{=}
\prod_{k=1}^M \left( 1- \epsilon_k \sqrt{\beta_{\frac12,\frac{N-k}{2}}} \right)
\end{equation}
where all variables are independent, and $\mathbb{P} (\epsilon_k = 1) = \mathbb{P} (\epsilon_k = -1) = \frac12$. The moment formula \eqref{truncation_at_edge_beta_1} follows from this equality in distribution; the computation of each factor is explained in \cite{BHNY08}.

\subsubsection*{Truncated unitary-symplectic case.}
A proof can be given along the same lines as above for unitary-symplectic matrices, with the following changes. We follow the convention used in \cite{DeiftGioev}; especially we take
\begin{equation}
J := \mathrm{diag} \left( 
\left(
\begin{array}{cc}
0 & 1 \\
-1 & 0
\end{array}
\right)
, \dots,
\left(
\begin{array}{cc}
0 & 1 \\
-1 & 0
\end{array}
\right)
\right) 
\end{equation}
and we think of unitary-symplectic matrices as composed of $2 \times 2$ blocks; or equivalently, as matrices of quaternions, unitary in the sense of quaternions. Columnwise, such a matrix can be written
\begin{equation}
U = (u_1 \ \| \ -J\overline{u}_1 \ \| \ \cdots \ \| \ u_N \ \| \ -J\overline{u}_N)
\end{equation}
and if $U$ is distributed according to the Haar measure on $\mathrm{Sp}(2N)$, then $u_1$ is uniformly distributed on the sphere $\mathcal{S}_{\mathbb{C}}^{2N-1}$. It follows from this and the definition of the Haar measure that the decomposition \eqref{rec_distr_CUE} can be replaced by
\begin{equation}\label{rec_distr_TUSE}
V_{2N}
\stackrel{d}{=} Q
\left(
\begin{array}{cc}
I_2 & 0 \\
0 & V_{2(N-1)}
\end{array}
\right)
\end{equation}
where $Q$ is a unitary-symplectic matrix with first two columns distributed as described above. The matrix $Q$ can be built from its first two columns as a quaternionic reflection, and the computation goes the same, by $2 \times 2$ blocks. Therefore, in the end of the calculation, it is not a scalar $v_1$ that factorizes out, but the determinant of a block, corresponding to a quaternionic norm:
\begin{equation}
\det \left( I_2 - \left(
\begin{array}{cc}
a+ib & -c+id \\
c+id & a-ib
\end{array}
\right) \right)
= (1-a)^2 + b^2 + c^2 + d^2.
\end{equation}
The conclusion is that, for $\beta=4$,
\begin{equation}
\det\left( I_{2M} - A \right)
\stackrel{d}{=}
\prod_{k=1}^M \left( (1-a_k)^2 + b_k^2 + c_k^2 + d_k^2 \right)
\end{equation}
where factors are independent, and $(a_k,b_k,c_k,d_k)_{k}$ have an explicit distribution, which is clear from the fact that $(a+ib, c+id)$ are the first two coefficients of a uniform unit vector. The moment formula \eqref{truncation_at_edge_beta_4} can be computed from this equality in distribution, as explained in \cite{BourgadeNikeghbaliRouault2011}, which also extends the same result (for $M=N$) to a wider range of compact groups.

\bibliography{bibliography}
\bibliographystyle{siam}

\end{document}